\documentclass[letterpaper, 10 pt, conference]{ieeeconf}  

\IEEEoverridecommandlockouts                              

\makeatletter

\let\proof\@undefined
\let\endproof\@undefined
\makeatother
\usepackage{tikz}
\usetikzlibrary{arrows}
\usetikzlibrary{shapes}
\usepackage{graphicx}
\usetikzlibrary{decorations.pathreplacing}
\usetikzlibrary{arrows,positioning}

\makeatletter

\pgfdeclareshape{quad with diagonal fill}
{
    \inheritsavedanchors[from=rectangle]
    \inheritanchorborder[from=rectangle]
    \inheritanchor[from=rectangle]{north}
    \inheritanchor[from=rectangle]{north west}
    \inheritanchor[from=rectangle]{north east}
    \inheritanchor[from=rectangle]{center}
    \inheritanchor[from=rectangle]{west}
    \inheritanchor[from=rectangle]{east}
    \inheritanchor[from=rectangle]{mid}
    \inheritanchor[from=rectangle]{mid west}
    \inheritanchor[from=rectangle]{mid east}
    \inheritanchor[from=rectangle]{base}
    \inheritanchor[from=rectangle]{base west}
    \inheritanchor[from=rectangle]{base east}
    \inheritanchor[from=rectangle]{south}
    \inheritanchor[from=rectangle]{south west}
    \inheritanchor[from=rectangle]{south east}

    \inheritbackgroundpath[from=rectangle]
    \inheritbeforebackgroundpath[from=rectangle]
    \inheritbehindforegroundpath[from=rectangle]
    \inheritforegroundpath[from=rectangle]
    \inheritbeforeforegroundpath[from=rectangle]

    \behindbackgroundpath{%
        \pgfextractx{\pgf@xa}{\southwest}%
        \pgfextracty{\pgf@ya}{\southwest}%
        \pgfextractx{\pgf@xb}{\northeast}%
        \pgfextracty{\pgf@yb}{\northeast}%
            \def\pgf@diagonal@point@a{\southwest}%
            \def\pgf@diagonal@point@b{\northeast}%
          \pgfpathmoveto{\pgfpointorigin}%
        \pgfpathlineto{\northeast}%
        \pgfpathlineto{\pgfpoint{\pgf@xb}{\pgf@ya}}%
        \pgfpathclose
        \color{\pgf@diagonal@left@color}%
        \pgfusepath{fill}%

          \pgfpathmoveto{\pgfpointorigin}%
        \pgfpathlineto{\pgf@diagonal@point@a}%
        \pgfpathlineto{\pgfpoint{\pgf@xa}{\pgf@yb}}%
        \pgfpathclose
        \color{\pgf@diagonal@left@color}%
        \pgfusepath{fill}%

        \pgfpathmoveto{\pgfpoint{\pgf@xa}{\pgf@yb}}%
        \pgfpathlineto{\pgfpointorigin}%
        \pgfpathlineto{\pgf@diagonal@point@b}%
        \pgfpathclose
        \color{\pgf@diagonal@top@color}%
        \pgfusepath{fill}%

        \pgfpathmoveto{\pgfpoint{\pgf@xb}{\pgf@ya}}%
        \pgfpathlineto{\pgfpointorigin}%
        \pgfpathlineto{\pgf@diagonal@point@a}%
        \pgfpathclose
        \color{\pgf@diagonal@top@color}%
        \pgfusepath{fill}%
    }
}

\newif\ifpgf@diagonal@lefttoright
\def\pgf@diagonal@top@color{white}
\def\pgf@diagonal@left@color{gray!30}

\def\pgfsetnscolor#1{\def\pgf@diagonal@top@color{#1}}%
\def\pgfsetewcolor#1{\def\pgf@diagonal@left@color{#1}}%

\tikzoption{ns color}{\pgfsetnscolor{#1}}
\tikzoption{ew color}{\pgfsetewcolor{#1}}
\makeatother

\usepackage{amsmath, amssymb}
\usepackage{amsthm}
\usepackage{graphicx}
\usepackage{algorithmic}
\usepackage{algorithm}
\usepackage{proof}
\usepackage{textcomp}

\usepackage{stmaryrd}

\newcommand{\stam}[1]{}
\newcommand{\comment}[1]{}

\newcommand{\G}{{\bf G}}
\newcommand{\F}{{\bf F}}
\newcommand{\X}{{\bf X}}
\newcommand{\U}{{\bf U}}

\newcommand{\rabin}{\mathcal{R}}
\newcommand{\mdp}{\mathcal{M}}
\newcommand{\pmdp}{\mathcal{P}}

\newcommand{\Pro}{\mathcal{P}}
\newcommand{\llb}{\llbracket}
\newcommand{\rrb}{\rrbracket}
\DeclareMathOperator*{\argmax}{arg\,max}

\newcommand{\pbar}{\bar{\pi}}
\newcommand{\good}{\mathcal{G}_i}
\newcommand{\bad}{\mathcal{B}_i}
\newcommand{\neutral}{\mathcal{N}_i}
\newcommand{\trans}{T_{\bar{\pi}}}
\newcommand{\recc}{R_{\bar{\pi}}}
\newcommand{\recco}{R_{\pi^*}}

\newcommand{\Uo}{U_{\pi^*}}
\newcommand{\Ubar}{U_{\bar{\pi}}}

\newcommand{\Wtb}{{\bf W}^{\text{tr}}}
\newcommand{\Wr}{{\bf W}^{\text{rec}}}
\newcommand{\Ut}{U^{\text{tr}}}
\newcommand{\Utb}{{\bf U}^{\text{tr}}}
\newcommand{\Ur}{{\bf U}^{\text{rec}}}

\newcommand{\Urso}{U_{\pi^*}^{\text{rec}}(s)}

\newtheorem{thm}{Theorem}
\newtheorem{definition}{Definition}
\newtheorem{prop}{Proposition}

\title{\LARGE \bf
A Learning Based Approach to Control Synthesis of Markov Decision Processes for Linear Temporal Logic Specifications
}

\author{
\thanks{
This work is supported in part by NDSEG and NSF Graduate Research Fellowships,
NSF grant CCF-1116993 and DOD ONR Office of Naval Research  N00014-13-1-0341. 
} \thanks{The authors are with the Department of Electrical Engineering and Computer Sciences, University of California, Berkeley, \tt{\{dsadigh, eskim, scoogan, sseshia, sastry\}@eecs.berkeley.edu.}} Dorsa Sadigh, Eric S. Kim, Samuel Coogan, S. Shankar Sastry, Sanjit A. Seshia\\
}

\begin{document}

\maketitle
\thispagestyle{empty}
\pagestyle{empty}

\begin{abstract}
We propose to synthesize a control policy for a Markov decision process (MDP) such that the resulting traces of the MDP satisfy a linear temporal logic (LTL) property. We construct a product MDP that incorporates a deterministic Rabin automaton generated from the desired LTL property. The reward function of the product MDP is defined from the acceptance condition of the Rabin automaton. This construction allows us to apply techniques from learning theory to the problem of synthesis for LTL specifications even when the transition probabilities are not known \emph{a priori}. We prove that our method is guaranteed to find a controller that satisfies the LTL property with probability one if such a policy exists, and we suggest empirically with a case study in traffic control that our method produces reasonable control strategies even when the LTL property cannot be satisfied with probability one.

\end{abstract}

\section{Introduction}
\label{sec:intro}
Control of Markov Decision Processes (MDPs) is a problem that is well
studied for applications such as robotics surgery,
unmanned aircraft control and  control of autonomous
vehicles~\cite{Alterovitz2007,Temizer2010,sadigh2014}. In recent
years, there has been an increased interest in exploiting the
expressiveness of temporal logic specifications in controlling
MDPs~\cite{Wolff2012,Lahijanian2009,Ding2014}.  
Linear Temporal Logic (LTL) provides a natural framework for
expressing rich properties such as stability, surveillance, response,
safety and liveness.  
Traditionally, control synthesis for LTL specifications is solved by finding a winning policy for a game
between system requirements and environment
assumptions~\cite{Piterman2006,Wongpiromsarn2012}. 

More recently, there has been an effort in exploiting these techniques
in designing controllers to satisfy high level specifications for
probabilistic systems. Ding \emph{et al.}~\cite{Ding2014} address this
problem by proposing an approach for finding a policy that maximizes
satisfaction of LTL specifications of the form $\phi = {\bf G F }\pi
\wedge \psi$ subject to minimization of the expected cost in between
visiting states satisfying $\pi$. In order to maximize the
satisfaction probability of $\phi$, the authors appeal to results from
probabilistic model checking~\cite{Baier2008,Vardi1999}. The methods
used for maximizing this probability take advantage of computing
\emph{maximal end components}, which are not well suited for partial
MDPs with unknown probabilities. We present a different technique that
does not require preprocessing of the model.  
Our algorithm learns the transition probabilities of a partial model
online. Our method can therefore be applied in practical contexts
where we start from a partial model with
unspecified probabilities.

Our approach is based on finding a policy that maximizes the expected
utility of an auxiliary MDP constructed from the original MDP and a
desired LTL specification. 
As in the above mentioned existing work, we convert the LTL
specification to a \emph{deterministic Rabin automaton} (DRA)~\cite{ltl2dstar2006,Safra1988}, and construct a product MDP such that
the states of the product MDP are pairs representing states of the
original MDP in addition to states of the DRA that encodes the desired
LTL specification. 
The novelty of our approach is that we then define a state based reward function on this product MDP based on
the \emph{Rabin} acceptance condition of the DRA.  
We extend our results to allow unknown transition probabilities and
learn them online. 
Furthermore, we select the reward function on the product MDP so it
corresponds to the \emph{Rabin} acceptance condition of the LTL
specification. 
Therefore,
any learning algorithm that optimizes the expected utility can be
applied to find a policy that satisfies the specification. 

We implement our method using a reinforcement learning algorithm that
finds the policy optimizing the expected utility of every state in the
Rabin-weighted product MDP.  
Moreover, we prove that if there exists a strategy that satisfies the
LTL specification with probability one, our method is guaranteed to
find such a strategy. 
For situations where a policy satisfying the LTL specification with probability one does not exist, our method finds reasonable strategies. 
We show this performance for two case studies: 1) Control of an
agent in a grid world, and 2) Control of a traffic network with
intersections. 

This paper is organized as follows: In Section~\ref{sec:prelim}, we
review necessary preliminaries. In Section~\ref{sec:formal}, we define
the synthesis problem and provide theoretical guarantees in finding a
policy satisfying the specification for a special
case. Section~\ref{sec:learning} discusses a learning approach
towards finding an optimal controller. We provide two case studies in
Section~\ref{sec:expts}. Finally, we conclude in
Section~\ref{sec:conclusion}. 

\section{Preliminaries}
\label{sec:prelim}
We introduce preliminaries on the specification language and the probabilistic model of the system.
We use Linear Temporal Logic (LTL) to define desired specifications. A LTL formula is built of \emph{atomic propositions} $\omega \in \Pi$ that are over states of the system that evaluate to $\texttt{True}$ or $\texttt{False}$, \emph{propositional formulas} $\phi$ that are composed of atomic propositions and Boolean operators such as $\wedge$ (and), $\neg$ (negation), and \emph{temporal operations} on $\phi$. Some of the common temporal operators are defined as:
\begin{equation*}
\begin{array}{l l l}
\G \phi& \phi \text{ is true all future moments.}\\
\F \phi& \phi \text{ is true some future moments.}\\
\X \phi& \phi \text{ is true the next moment.}\\
\phi_1 \U \phi_2& \phi_1 \text{ is true until $\phi_2$ becomes true.}
\end{array}
\end{equation*}


Using LTL, we can define interesting \emph{liveness} and \emph{safety} properties such as surveillance properties $\G \F \phi$, or stability properties $\F \G \phi$. 

\begin{definition}
\label{def:DRA}
A \emph{deterministic Rabin automaton} is a tuple $\rabin = \langle Q, \Sigma, \delta, q_0, F \rangle$ where $Q$ is the set of states; $\Sigma$ is the input alphabet; $\delta: Q \times \Sigma \rightarrow Q$ is the transition function; $q_0$ is the initial state and $F$ represents the acceptance condition: $F = \{(G_1, B_1), \dotsc, (G_{n_F}, B_{n_F})\}$ where $G_i, B_i \subset Q$ for $i=1,\dotsc, n_F$.
\end{definition}

A \emph{run} of a Rabin automaton is an infinite sequence $r = q_0q_1\dotsc$ where $q_0 \in Q_0$ and for all $i > 0$,  $q_{i+1} \in \delta(q_i, \sigma)$, for some input $\sigma \in \Sigma$.
For every run $r$ of the Rabin automaton, $\text{inf}(r) \in Q$ is the set of states that are visited infinitely often in the sequence $r = q_0q_1\dotsc$. 
A run $r=q_0q_1\dotsc$ is \emph{accepting} if there exists $i \in \{1,\dotsc ,n_F\}$ such that:
\begin{equation}
\text{inf}(r) \cap G_i \neq \emptyset \quad \text{and} \quad \text{inf}(r) \cap B_i = \emptyset
\end{equation}

For any LTL formula $\phi$ over $\Pi$, a deterministic Rabin automaton (DRA) can be constructed with input alphabet $\Sigma = 2^\Pi$ that accepts all and only words over $\Pi$ that satisfy $\phi$~\cite{Safra1988}. We let $\mathcal{R}_\phi$ denote this DRA.

\begin{definition}
\label{def:MDP}
A labeled \emph{Markov Decision Process (MDP)} is a tuple $\mdp = \langle S, \mathcal{A}, P, s_0, \Pi, L\rangle$ where $S$ is a finite set of states of the MDP; $A$ is a finite set of possible actions (controls) and $\mathcal{A}: S\rightarrow 2^A$ is defined as the mapping from states to actions; $P$ is a transition probability function defined as $P: S \times A \times S \rightarrow [0,1]$; $s_0 \in S$ is the initial state;
$\Pi$ is a set of atomic propositions, and $L: S\rightarrow 2^\Pi$ is a labeling function that labels a set of states with atomic propositions.  
\end{definition}


\section{Synthesis through Reward Maximization}
\subsection{Problem Formulation}

\label{sec:formal}

Consider a labeled MDP 
\begin{equation}
  \label{eq:1}
\mdp = \langle S, \mathcal{A}, P, s_0, \Pi, L\rangle
\end{equation}
and a linear temporal logic specification $\phi$. 

\begin{definition}
A \emph{policy} for $\mdp$ is a function $\pi:S^+\to A$ such that $\pi(s_0 s_1\ldots s_n)\in\mathcal{A}(s_n)$ for all $s_0s_1\ldots s_n\in S^+$ where $S^+$ denotes the set of all finite sequences of states in $S$. 
\end{definition}
Observe that a policy $\pi$ for an MDP $\mdp$ induces a Markov chain which we denote by $\mdp_{\pi}$.  A run of a Markov chain is an infinite sequence of states $s_0,s_1,\dotsc$, where $s_0$ is the initial state of the Markov chain, and for all $i, P(s_i,a,s_{i+1})$ is nonzero for some action $a \in A$.

Our objective is to compute a policy $\pi^*$ for $\mdp$ such that the runs of $\mdp_{\pi^*}$ satisfy the LTL formula $\phi$ with probability one as defined below.  Our approach composes $\mdp$ and the DRA $\rabin_\phi=\langle Q, \Sigma, \delta, q_0, F \rangle$ whose acceptance condition corresponds to satisfaction of  $\phi$. We then obtain a policy $\pi^*$ for this composition. Our approach is particularly amenable to learning-based algorithms as we discuss in Section \ref{sec:learning}. In particular, the policy $\pi^*$ can be constructed even when the transition probabilities $P$ for $\mdp$ are not known. Thus, we present an approach that allows the policy $\pi^*$ to be found \emph{online} while learning the transition probabilities of $\mdp$.

We create a \emph{Rabin weighted product MDP} $\pmdp$ ,defined below, using the DRA $\rabin_\phi$ and labeled MDP $\mdp$. The set of states $S_\Pro$ in  $\pmdp$ are a set of augmented states with components that correspond to states in $\mdp$ and components that correspond to states in $\rabin_\phi$. The set of actions $\mathcal{A}_\Pro$ is identical to the set of actions in $\mdp$. 

To this end, we define a \emph{Rabin weighted product MDP} given a MDP $\mdp$ and a DRA $\rabin$ as follows:

\begin{definition}
\label{def:PMDP}
A \emph{Rabin weighted product MDP}  or simply \emph{product MDP}
between a labeled MDP $\mdp = \langle S, \mathcal{A}, P, s_0, \Pi, L\rangle$ and a DRA $\rabin = \langle Q, \Sigma, \delta, q_0, F \rangle$ is defined as a tuple $\pmdp = \langle S_\Pro, \mathcal{A}_\Pro, P_\Pro, s_{\Pro 0}, F_\Pro, W_\Pro \rangle$~\cite{Ding2014}, where:

\begin{itemize}
\item $S_\Pro = S \times Q$ is the set of states.
\item $\mathcal{A}_\Pro$ provides the set of control actions from the MDP:  $\mathcal{A}_\Pro ((s,q)) = \mathcal{A} (s)$.
\item $P_\Pro$ is the set of transition probabilities defined as:
\begin{equation}
P_\Pro(s_\Pro,a, s'_\Pro) =
  \begin{cases}
   P(s,a,s') & \text{if } q' = \delta(q, L(s)) \\
   0      & \text{otherwise }
  \end{cases}
\end{equation}
$s_\Pro = (s,q) \in S_\Pro$ and $s'_\Pro = (s',q')$. 
\item $s_{\Pro 0}=(s_0,q_0) \in S_\Pro$ is the initial state,
\item $F_\Pro$ is the acceptance condition given by
 \[
 F_\Pro = \{ (\mathcal{G}_1,\mathcal{B}_1), \dotsc, (\mathcal{G}_{n_F}, \mathcal{B}_{n_F})\}
 \]
 where $\mathcal{G}_i=S\times G_i$ and $\mathcal{B}_i=S\times B_i$.
\item For the above acceptance condition, $W_\Pro = \{ W_\Pro^i\}_{i=1}^{n_F}$ is a collection of reward functions $W_\Pro^i:S_\Pro\to\mathbb{R}$ defined by:
\begin{equation}
W^i_\Pro(s_\Pro) =
  \begin{cases}
   w_G & \text{if } s_\Pro \in \good \\
   w_B & \text{if } s_\Pro \in \bad \\
   0      & \text{if } s_\Pro \in S \backslash \big(\good \cup \bad\big)  \\
  \end{cases}
\end{equation}
where $w_G > 0$ is a positive reward, $w_B <0 $ is a negative reward.\\
We let $\mathcal{N}_i = S \backslash \big ( \good \cup \bad \big)$ for every pair of $(\good,\bad)$. 
\end{itemize}
\end{definition}

We use the notation $\Pro^i$ to denote $\Pro$ with the specific reward function $W^i_\Pro$.  In seeking a policy $\pi$ for $\mdp$ such that $\mdp_\pi$ satisfies $\phi$, it suffices to consider \emph{stationary policies} of the corresponding Rabin weighted product MDP~\cite{Baier2008}.

\begin{definition}
A \emph{stationary policy} $\pi$ for a product MDP $\pmdp$ is a mapping $\pi: S_\Pro \rightarrow A_\Pro$ that maps every state to actions selected by policy $\pi$.
\end{definition}
A stationary policy for $\pmdp$ corresponds to a finite memory policy for $\mdp$.
We let $\pmdp_{\pi}$ denote the Markov chain induced by applying the stationary policy $\pi$ to the product MDP ${\pmdp}$. 
Let $r = s_{\Pro 0}s_{\Pro1}s_{\Pro2}\ldots$ be a run of ${\Pro_\pi}$ with initial product state $s_{\Pro 0}$.



\begin{definition}
\label{def:accept}
Consider a MDP $\mdp$ and a LTL formula $\phi$ with corresponding DRA $\rabin_\phi$, let $\pmdp$ be the corresponding Rabin weighted MDP, and let $\pi$ be a stationary policy on $\pmdp$. We say that $\mdp_\pi$ \emph{satisfies $\phi$ with probability $1$} if 
\begin{align*}
&Pr(\{ r:   \exists (\good,\bad) \in F_\Pro(s) \\
&\text{inf} (r) \cap \good \neq \emptyset
  \: \wedge \: \text{inf}(r) \cap \bad = \emptyset  \}) = 1
\end{align*}
where $r$ is a run of $\Pro_{\pi}$ initialized at $s_{\Pro 0}$. 
\end{definition}

Intuitively, $\mdp_\pi$ satisfies $\phi$ with probability one if the probability measure of the runs of ${\Pro_\pi}$ that violate the acceptance condition of $\phi$ is 0.\\
We let $i$ be \emph{index of Rabin} acceptance condition for property $\phi$.
A reward function $W_\Pro^i(s_\Pro)$ on every state is specified in Definition~\ref{def:PMDP} and can be identified by ${\bf W}^i \in \mathbb{R}^{|S_\Pro|}$ for some enumeration of $S_\Pro$. We assign a negative reward $w_B$ to states $s_\Pro \in \bad=S \times B_i$ since we would like to visit them only finitely often. Similarly we assign positive rewards $w_g$ to $s_\Pro \in \good$, and reward of $0$ on neutral states $s_\Pro \in \neutral$ to bias the policy towards satisfaction of the Rabin automaton's acceptance condition. 

\begin{definition}
\label{def:discountreward}
For $i \in \{ 1, \dotsc , n_F\}$, the \emph{expected discounted utility} for a policy $\pi$ on $\pmdp^i$ with discount factor $0 < \gamma < 1$ is a vector
${\bf U}^i_{\pi} = [U^i_{\pi}(s_0) \dotsc U^i_{\pi}(s_N)]$ for $s_k \in S_{\Pro}, k \in \{1,\dotsc, N\}$ and $N = |S_\Pro|$, such that:

\begin{align}
\textbf{U}^i_{\pi} = \sum_{n=0}^{\infty} \gamma^n P^n_{\pi}{\bf W}^i
\end{align}

where ${\bf W}^i$ is the vector of the rewards $W^i_{\Pro}(s_\Pro)$ and $P_\pi$ is a matrix containing the probabilities $P_\Pro(s_\Pro, \pi(s_\Pro), s'_\Pro)$.
For simpler notation, we omit the superscript $i$ the index of Rabin acceptance condition of the LTL specification. In the rest of this paper, it is assumed that ${\bf W}$ and ${\bf U}_{\pi}$ are the reward and utility vectors of the product MDP with their corresponding set of Rabin acceptance condition pair ($\mathcal{G}_i, \mathcal{B}_i$).
\end{definition}

\begin{definition}
\label{def:optimal}
A policy that maximizes this expected discounted utility for every state is an \emph{optimal policy} {\boldmath $\pi^*$} = $[ \pi^*(s_0) \dotsc \pi^*(s_N)]$, defined as:
\begin{align}
\text{{\boldmath $\pi^*$}} = \argmax_{\text{\boldmath{$ \pi$}}} \sum_{n=0}^{\infty} \gamma^n P^n_{\pi}{\bf W}
\end{align}
\end{definition}

Note that for any policy $\pi$, for all $ s \in S_\Pro \quad U_{\pi}(s) \leq U_{\pi^*}(s)$.
From a product MDP $\pmdp$, we seek a policy that satisfies the LTL specification by optimizing the expected future utility. 
Note that an optimal policy exists for each acceptance condition $(\mathcal{G}_i,\mathcal{B}_i)\in F_\Pro$ and thus our reward maximization algorithm must be run on each acceptance condition. The outcome is a collection of strategies $\{\pi_i^*\}_{i=1}^{n_F}$ where $\pi_i^*$ is the optimal policy under rewards $W_\Pro^i$. We use Definition~\ref{def:accept} to determine whether a policy $\pi_i^*$ satisfies $\phi$ with probability one by analyzing properties of the recurrent classes in $\pmdp$  ~\cite{Baier2008}.


The following theorem shows that optimizing the expected discounted utility produces a policy $\pi$ such that $\mdp_\pi$ satisfies $\phi$ with probability one if such a policy exists.


\begin{thm}
\label{thm:1}
Given MDP $\mdp$ and LTL formula $\phi$ with corresponding Rabin weighted product MDP $\Pro$. If there exists a policy $\pbar$ such that $\mdp_{\bar{\pi}}$ satisfies $\phi$ with probability $1$, then there exists $i^*\in\{1,\ldots,n_F\}$, $\gamma^*\in[0,1)$, and $w_B^*<0$ such that any algorithm that optimizes the expected future utility of $\Pro^{i^*}$ with $\gamma\geq \gamma^*$ and $w_B\leq w_B^*$ will find such a policy.
\end{thm}

\begin{proof}
Proof of theorem~\ref{thm:1} can be found in Appendix~\ref{appendix}. Intuitively, choosing $\gamma$ i.e. the discount factor close to $1$ enforces visiting $\good$ infinitely often, and a large enough negative reward $w_B$ enforces visiting $\bad$ only finitely often. This will result in satisfaction of $\phi$ by our algorithm.
\end{proof}

Theorem \ref{thm:1} provides a practical approach to synthesizing a control policy $\pi^*$ for the MDP $\mdp$. After constructing the corresponding product MDP $\Pro$, a collection of policies $\{\pi_i^*\}_{i=1}^{n_F}$ is computed that optimize the expected future utility of each $\Pro^i$. 
Provided that $\gamma$ and $|w_B|$ are sufficiently large,
if there exists a policy $\pi$ such that $\mdp_{\pi}$ satisfies $\phi$ with probability $1$,
then for at least one of the computed policies $\pi_i^*$, $\mdp_{\pi^*_i}$ satisfies $\phi$ with probability $1$.
Determining which of the policies satisfy $\phi$ with probability $1$ is easily achieved by computing strongly connected components of the resulting Markov chains, for which there exists efficient graph theoretic algorithms~\cite{Baier2008}.

In this section, we have not provided an explicit method for optimizing the expected utility of the product MDP $\Pro$.  If the transition probabilities of $\mdp$ are not known \emph{a priori}, then the optimization algorithm must simultaneously learn the transition probabilities while optimizing the expected utility, and tools from learning theory are well-suited for this task. In the following section, we discuss how these tools apply to the policy synthesis problem above.




\label{sec:Algorithm}
\subsection{Synthesis through Reinforcement Learning}
\label{sec:learning}

By translating the LTL synthesis problem into an expected reward maximization framework in section~\ref{sec:formal}, it is now possible to use standard techniques in the reinforcement learning literature to find satisfying control policies.

In the previous section, we did not provide an explicit method for optimizing the expected utility of the product MDP $\Pro$.  If the transition probabilities of $\mdp$ are not known \emph{a priori}, then the optimization algorithm must 1) Learn the transition probabilities and 2) Optimize the expected utility. Tools from learning theory are well-suited for this task. 

Algorithm~\ref{algo:ptd} below is a modified active temporal difference learning algorithm~\cite{Russell:2010xe} that accomplishes these goals. It is called after each observed transition and updates a set of persistent variables, which include a table of transition frequencies, state utilities, and the optimal policy that can each be initialized by the user with a priori estimates. The magnitude of the update is determined by a learning rate, $\alpha$. 
\begin{algorithm}[h]
\caption{Temporal Difference Learning for $\mathcal{M}_\Pro$}
\label{algo:ptd}  
\begin{algorithmic}
\STATE \textbf{Input:} $s_\Pro'$ Current state of $\pmdp$.
\STATE \textbf{Output:} $a'_\Pro$ Current action
\STATE \textbf{Persistent Values: }
\STATE $\cdot$ Utilities $U_\Pro (s_\Pro) $ for all states of $\pmdp$ initialized at $0$.
\STATE $\cdot$ $N_{sa} (\llb s_\Pro \rrb,a_\Pro) $ a table of frequency of state, action pairs initialized by the user.
\STATE $\cdot$ $N_{s'|sa} (\llb s_\Pro \rrb,a_\Pro,\llb s_\Pro' \rrb)$ a table of frequency of the outcome of the equivalence class $\llb s'_\Pro \rrb$ for state, action pairs in the equivalence class $(\llb s_\Pro \rrb,a_\Pro)$ initialized by the user.
\STATE $\cdot$ Optimal Policy $\pi^*$ for every state. Initialized at $0$.\\
\STATE $\cdot$ $s_\Pro, a_\Pro$ previous state and action, initialized as \texttt{null}
 
\IF {$s_\Pro'$ is new}
\STATE $U_\Pro(s_\Pro') \leftarrow W_{\Pro}^i(s_\Pro')$
\ENDIF
\IF {ResetConditionMet() is True}
\STATE $s'_\Pro$ = ResetRabinState($s'_\Pro$)
\ELSIF {$s_\Pro$ is not \texttt{NULL}}
\STATE
$N_{sa}(\llb s_\Pro \rrb, a_\Pro) \leftarrow N_{sa}(\llb s_\Pro \rrb, a_\Pro) + 1$\\
$N_{s'|sa}(\llb s_\Pro \rrb, a_\Pro, \llb s'_\Pro \rrb) \leftarrow N_{s'|sa}(\llb s_\Pro \rrb, a_\Pro, \llb s'_\Pro \rrb) + 1$\\
\FORALL {$t$ that $N_{s'|sa} (\llb s \rrb,a, \llb t \rrb) \neq 0$}
\STATE $P(\llb s \rrb,a,\llb t \rrb) \leftarrow$ \\ $\qquad\qquad N_{s'|sa}(\llb s \rrb,a, \llb t \rrb)/N_{sa}(\llb s \rrb, a)$
\ENDFOR

\STATE $U_\Pro(s_\Pro) \leftarrow $ \\ $\qquad\qquad \alpha \: U_\Pro(s_\Pro) + 
(1-\alpha) [ W_\Pro^i(s_\Pro) $ \\
$\qquad\qquad + \text{ } \gamma \max_{a} \sum_{\sigma} P(s_\Pro,a_\Pro,\sigma) U(\sigma)]$

\STATE $\pi^*(s_\Pro) \leftarrow \arg \max_{a \in A_\Pro(s_\Pro)} \sum_{\sigma} P(s_\Pro,a_\Pro,\sigma)U(\sigma)$

\ENDIF
\STATE Choose current action $a'_\Pro = f_{exp} $
\STATE $s_\Pro = s'_\Pro$
\STATE $a_\Pro = a'_{\Pro}$
\end{algorithmic}
\end{algorithm}
Algorithm~\ref{algo:ptd} is customized to take advantage of the structure in $\pmdp$ to converge more quickly to the actual transition probabilities. Observe that product states corresponding to the same labeled MDP state have the same transition probability structure i.e. $P_\Pro(s_\Pro, a, s'_\Pro) = P_\Pro(\hat{s}_\Pro, a, \hat{s}'_\Pro)$ if $s_\Pro = (s,q)$, $\hat{s}_\Pro  = (s,\hat{q})$, $\hat{s}_\Pro = (s', q')$, and $\hat{s}'_\Pro=(s', \hat{q}')$, where $q,q',\hat{q},\hat{q}' \in Q$, and $s,s' \in S$. Therefore, every iteration in the product MDP can in fact be used to update the transition probability estimates for all product MDP states that share the same labeled MDP state. Thus, the algorithm uses equivalence classes $(\llb s_\Pro \rrb,a_\Pro)$, where $\llb s_\Pro \rrb = s \times Q  = \{s_\Pro = (s,q) \:|\: q \in Q\}$ to more quickly converge to the optimal policy.


Traditionally, temporal difference learning occurs over multiple trials where the initial state is reset after each trial~\cite{sutton1988}. Similarly, in an \emph{online} application, where we cannot reset the labeled MDP state, we periodically reset the Rabin component of the product state to $Q_0$. For instance, if the LTL formula contains any safety specifications, then a safety violation will make it impossible to reach a state with positive reward in $\pmdp$. 
To ensure we obtain a correct control action for every state
we introduce a function ``ResetConditionMet()" in Algorithm~\ref{algo:ptd} that forces a Rabin state reset whenever a safety violation is detected, or heuristically after a set time interval if liveness properties are not being met.
In both case studies, we observed that this reset technique results in Algorithm~\ref{algo:ptd} converging to a satisfying policy.

We note that online learning algorithms on general MDPs do not have hard convergence guarantees to the optimal policy because of the exploitation versus exploration dilemma~\cite{Russell:2010xe}. A learning agent decides whether to explore or exploit via the exploration function $f_{exp}$.  One possible exploration function for \emph{probably approximately correct learning} observes transitions and builds an internal model of the transition probabilities. The agent defaults to an exploration mode and only explores if it can learn more about the system dynamics ~\cite{Kearns2002}.




\section{Case Studies}
\label{sec:expts}
\subsection{Control of an agent in a grid world}
\begin{figure}

\centering

\includegraphics[width=.8\linewidth]{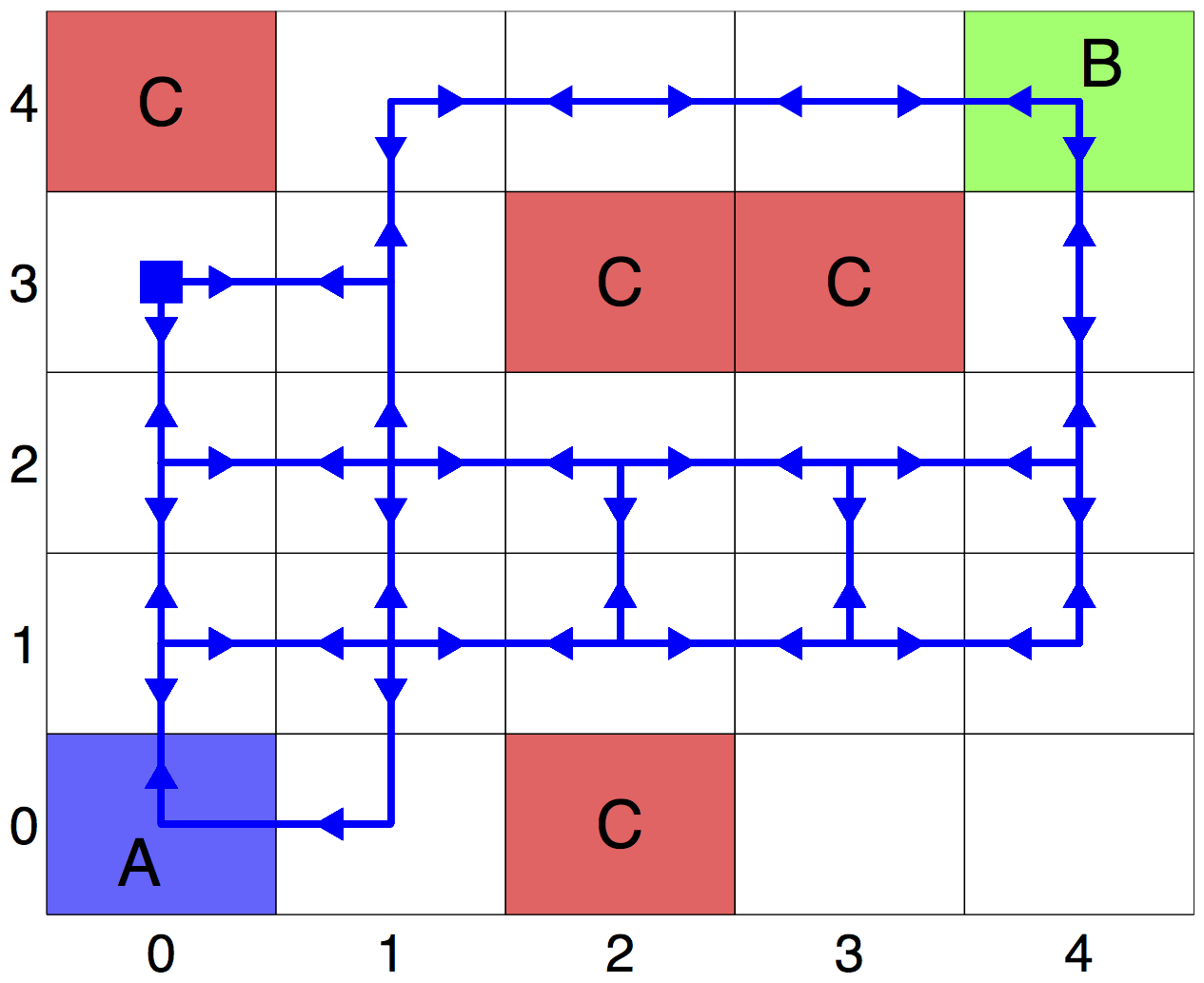}

\caption{A grid world example with a superimposed sample trajectory under the policy $\pi^*$ generated by the reinforcement learning algorithm. The trajectory has a length of 1000 time steps and an initial location (0,3) denoted by a solid square. The arrows denote movement from the box containing the arrow to a corresponding adjacent state. Locations (3,0) and (4,0) do not have any arrows because they are not reachable from the initial state under our policy. Note that $\pi^*$ is deterministic, but may cause a single location on the grid (e.g. location (4,2)) to have different actions under different Rabin states.}

\label{fig:gridworld}
\end{figure}
For illustrative purposes, we consider an agent in a $5\times5$ grid world that is required to visit regions labeled $A$ and $B$ infinitely often, while avoiding region $C$. The LTL specification is given as the following formula:
\begin{align}
\G\F A \wedge \G\F B \wedge \G \neg C
\end{align}

The agent is allowed four actions, where each one expresses a preference for a diagonal direction. An ``upper right" action will cause the agent to move \emph{right} with probability 0.4, \emph{up} with probability 0.4, and remain \emph{stationary} with probability 0.2. If a wall is located to the agent's right then it will move \emph{up} with probability 0.8, if one is located above then it will move to the \emph{right} with probability 0.8, and if the agent is in the upper right corner, then it is guaranteed to remain in the same location. The dynamics for the other actions are identical after an appropriate rotation. 

Figure~\ref{fig:gridworld} shows the results of the learning algorithm with an exploration function $f_{exp}(\cdot)$ that simply outputs random actions while learning. The product MDP contained 150 states and one acceptance pair, $\good = 500, \bad = -500$ and $\gamma = 0.98$. There were 600 trials, which are separated by a Rabin reset every 200 time steps. 

Observe that no policy exists such that $\phi$ is satisfied for all runs of the MDP. For example, it is possible that every action results in no movement of the robot. However, it is clear that there exists a policy that satisfies $\phi$ with probability 1, thus this example satisfies the conditions for Theorem~\ref{thm:1}. 



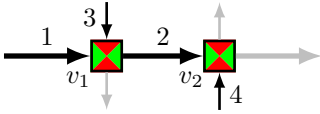
\begin{figure}
\centering
  \begin{tikzpicture}
    [subsystem/.style={quad with diagonal fill, ns color=red, ew color=green, text =black, draw,inner sep=1pt,minimum size=4mm},
     nostate/.style={gray!50!white},
>=to,line width=1pt, >=latex, scale=1]
\node[subsystem,label={[label distance=-.2cm]225:$v_1$}] (node2) at (1.5,0) {};
\node[subsystem,label={[label distance=-.2cm]225:$v_2$}] (node3) at (3,0) {};
\node (A) at (0,0) {};
\node (B) at (4.5,0) {};
\draw[->, line width=2pt] (A) -- node[above]{$1$}(node2);
\draw[->, line width=2pt] (node2) -- node[above]{$2$} (node3);
\draw[->, nostate, line width=2pt] (node3) -- (B);
\draw[<-] (node2.90) -- node[left,pos=.6]{$3$} +(0,.5);
\draw[->,nostate] (node2.270) -- +(0,-.5);
\draw[->, nostate] (node3.90) -- +(0,.5);
\draw[<-] (node3.270) -- node[right,pos=.6]{$4$} +(0,-.5);
\end{tikzpicture}
  \caption{A traffic network consisting of East-West links 1 and 2 and North-South links 3 and 4 and two signalized intersections. The gray links are not explicitly modeled.}
  \label{fig:traffic}
\end{figure}

\subsection{Control of a Traffic Network with Two Intersections}

To demonstrate the utility of our approach, we apply our control synthesis algorithm to a traffic network with two signalized intersections as depicted in Figure~\ref{fig:traffic}. 
We employ a traffic flow model with a time step of 15 seconds.
At each discrete time step, signal $v_1$ either actuates link $1$ or link $3$, and signal $v_2$ actuates link $2$ or link $4$. For $i=1,2$, the Boolean variable $s_{v_i}$ is equal to 1 if link $i$ is actuated at signal $v_i$ and is equal to 0 otherwise.  The set of control actions is then
\begin{align}
\label{eq:cont}
  A\triangleq\{(1,2), (1,4), (3,2), (3,4)\}
\end{align}
where, for $a\in A$, $l\in a$ implies that link $l$ is actuated.
The gray links in Fig. \ref{fig:traffic} are not explicitly considered in the model as they carry traffic out of the network. 

The model considers a queue of vehicles waiting on each link, and at each time step, the queue is forwarded to downstream links if the queue's link is actuated and if there is available road space downstream. If the queue is longer than some \emph{saturating limit}, then only this limit is forwarded and the remainder remains enqueue for the next time step. The vehicles that are forwarded divide among downstream links via \emph{turn ratios} given with the model.  

Let $C_l>0$ be the capacity of link $l$. 
Here, the queue length is assumed to take on continuous values. To obtain a discrete model, the interval $[0,C_l]\subset \mathbb{R}$ is divided into a finite, disjoint set of subintervals. For example, if link $l$ can accomodate up to $C_l=40$ vehicles, we may divide $[0,40]$ into the set $\{[0,10],(10,20],(20,30],(30,40]\}$. The current discrete state of link $l$ is then the subinterval that contains the current queue length of link $l$, and the total state of the network is the collection of current subintervals containing the current queue lengths of each link. 


Here, we consider probabilistic transitions among the discrete states and obtain an MDP model with control actions $A$ as defined in \eqref{eq:cont}. For the example in Fig. \ref{fig:traffic}, we have
\begin{align}
  (C_1,C_2,C_3,C_4)=(40,50,30,30)
\end{align}
and link 1 is divided into four subintervals, link 2 is divided into five subintervals, and links 3 and 4 are divided into two subintervals each. In addition, we augment the state space with the last applied control action so that the control objective, expressed as a LTL formula, may include conditions on the traffic lights as is the case below, thus there are 320 total discrete states. The transition probabilities for the MDP model are determined by the specific subintervals, saturating limits, and turn ratios. Future research will investigate the details of abstracting the traffic dynamics to an MDP.

Let $x_i$ for $i=1,\ldots, 4$ denote the number of vehicles enqueue on link $i$. We consider the following control objective:
\begin{align}
\label{eq:ob1}
&  \F \G(x_1\leq 30 \land  x_2\leq 30) \land\\
&\G \F(x_3\leq 10)\land \G \F (x_4\leq 10) \land \\
\label{eq:ob3}&\G ((s_{v_2}\land \X(\lnot s_{v_2}))\implies(\X\X(\lnot s_{v_2})\land \X\X\X(\lnot s_{v_2}))).
\end{align}
In words, \eqref{eq:ob1}--\eqref{eq:ob3} is
\begin{align*}
&\text{(Eventually links 1 and 2 have adequate supply) and}\\
&\text{(Infinitely often, links 3 and 4 have short queues) and}\\
&\text{(When signal $v_2$ actuates link 4,}\\
&\text{it does so for a minimum of 3 times steps)}
\end{align*}
where ``adequate supply'' means the number of vehicles on links $1$ and $2$ does not exceed 30 vehicles and thus can always accept incoming traffic, and a queue is ``short'' if the queue length is less than 10. Condition \eqref{eq:ob3} is a minimum green time for actuation of link 4 at signal 2 and may be necessary if, \emph{e.g.}, there is a pedestrian crosswalk across link 2 which requires at least 45 seconds (three time steps) for safe crossing (recall that $s_{v_2}=1$ when link $2$ is actuated). The above condition is encoded in a Rabin automaton with one acceptance pair and 37 states. The Rabin-weighted product MDP contains 11,840 states and rewards corresponding to the one acceptance pair.

In Fig. \ref{fig:traffic2}, we explore how our approach can be used to synthesize a control policy. Restating \eqref{eq:ob1}--\eqref{eq:ob3}, the control objective requires the two solid traces to eventually remain below the threshold at 30 vehicles and for the two dashed traces to infinitely often move below the threshold at 10 vehicles. Additionally, signal $2$ should be red for at least three consecutive time steps whenever it switches from green to red.

  Fig. \ref{fig:traffic2}(a) shows a na\"ive control policy that synchronously actuates each link for 3 time steps but does not satisfy $\phi$ since $x_2$ remains above 30 vehicles. If estimates of turn ratios and saturation limits are available from, \emph{e.g.}, historical data, then we can obtain a MDP that approximates the true traffic dynamics and determine the optimal control policy for the corresponding Rabin-weighted product MDP. When applied to the true traffic model, the controller greatly outperforms the naive policy but still does not satisfy $\phi$, as shown in  Fig. \ref{fig:traffic2}(b). However, by modifying this policy via reinforcement learning on the true traffic dynamics, we obtain a controller that empirically often satisfies $\phi$ as seen in Fig. \ref{fig:traffic2}(c) (Note that we should not expect $\phi$ to be satisfied for all traces of the MDP or all disturbance inputs as such a controller may not exist).

This example suggests how our approach can be utilized in practice: a ``reasonable'' controller can be obtained by using a Rabin-weighted MDP generated from approximated traffic parameters. This policy can then be modified \emph{online} to obtain a control policy that better accommodates existing conditions. Additionally, using a suboptimal controller prior to learning is rarely of serious concern for traffic control as the cost is only increased delay and congestion.
\begin{figure}[!]
  \centering
  \includegraphics[width=\linewidth]{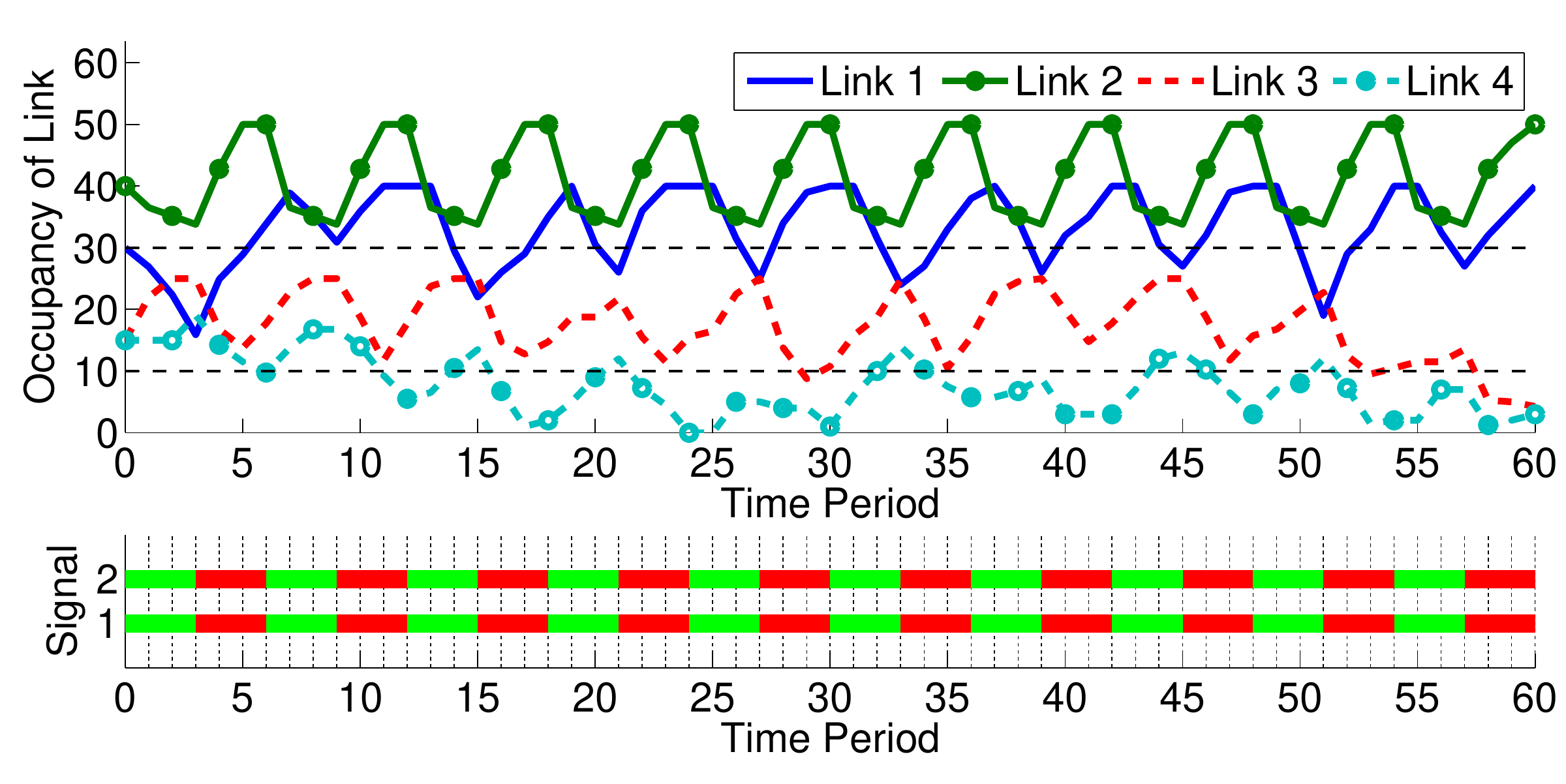}\\
(a)\\
  \includegraphics[width=\linewidth]{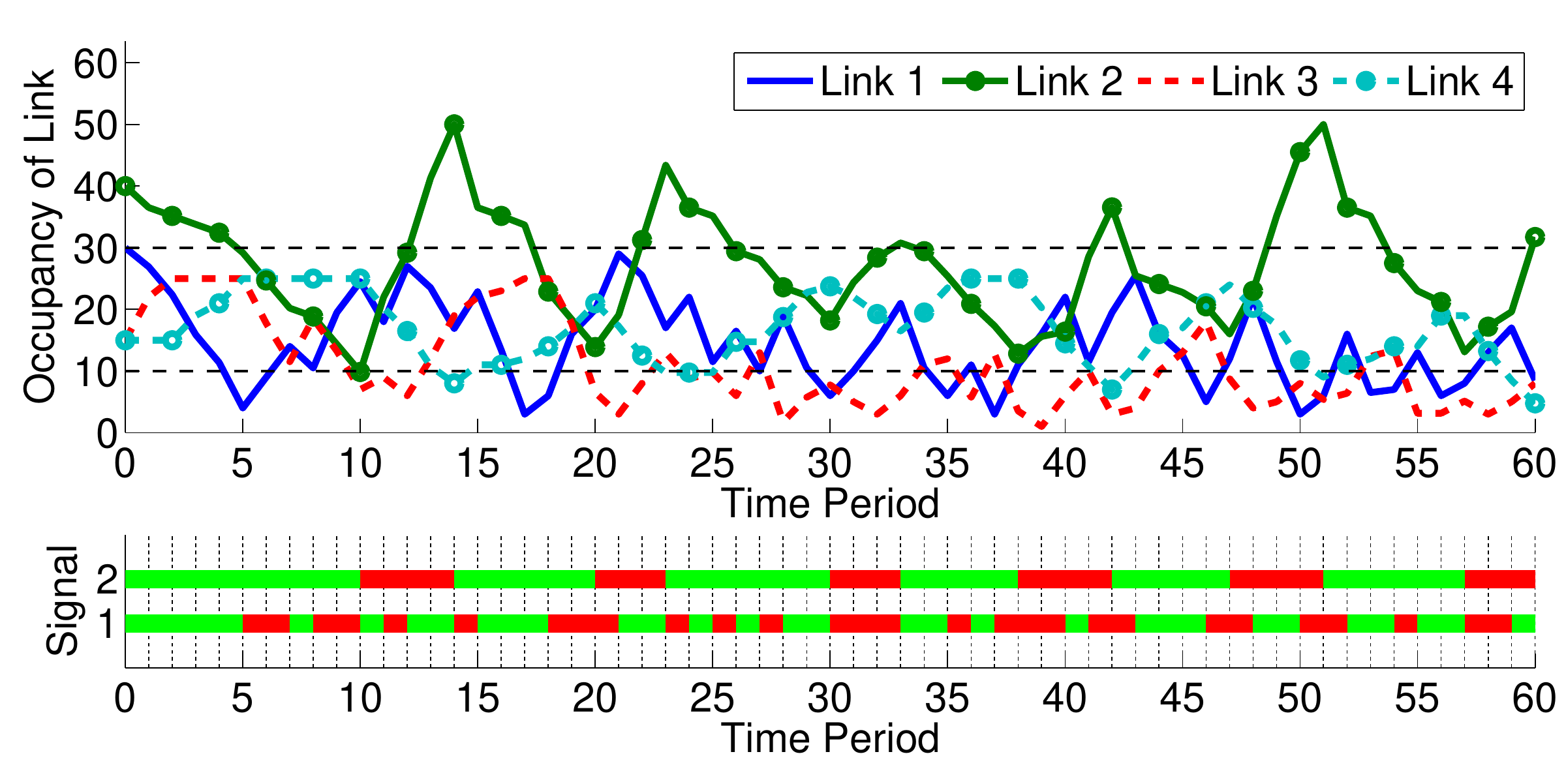}\\
(b)\\
  \includegraphics[width=\linewidth]{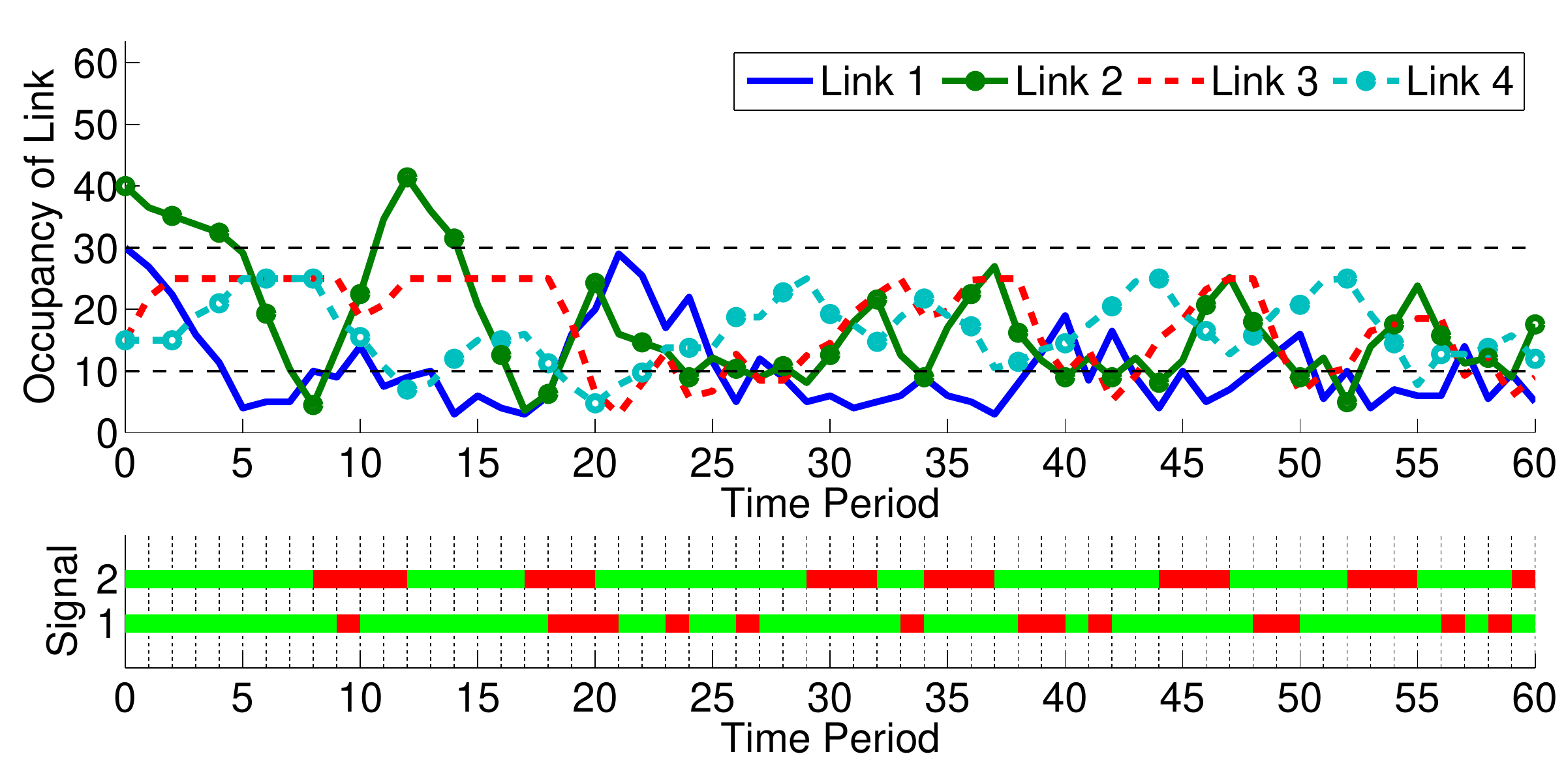}\\
(c)
  \caption{Sample trajectories of the traffic network in Fig. \ref{fig:traffic}.  \textbf{(a)} A simple controller that synchronously actuates links for 3 time periods and does not satisfy $\phi$. \textbf{(b)} An optimal controller for an MDP obtained from an approximate model of the traffic dynamics (\emph{e.g.}, a model with turn ratios and saturation limits different than reality).  This controller outperforms the previous na\"ive controller, but does not fully satisfy $\phi$. \textbf{(c)} The controller from (b) is modified via reinforcement learning on the true traffic model. In the lower plot for all cases, signal $i$ for $i=1,2$ is green if link $i$ is actuated and is red otherwise.  This example suggests how a reasonable control policy can be obtained from an approximate MDP estimated via, \emph{e.g.}, historical data and modified ``online'' using reinforcement learning on observed traffic dynamics.}
  \label{fig:traffic2}
\end{figure}




\section{Conclusion}
\label{sec:conclusion}
We have proposed a method for synthesizing a control policy for a MDP such that traces of the MDP satisfy a control objective expressed as a LTL formula. We proved that our synthesis method is guaranteed to return a controller that satisfies the LTL formula with probability one if such a controller exists. We provided two case studies: In the first case study, we utilize the proposed method to synthesize a control policy for a virtual agent in a gridded environment, and in the second case study, we synthesize a traffic signal controller for a small traffic network with two signalized intersections.

The most immediate direction for future research is to investigate theoretical guarantees in the case when the LTL specification cannot be satisfied with probability one. For example, it is desirable to prove or disprove the conjecture that for appropriate weightings in the reward function, our proposed method finds the control policy that maximizes the probability of satisfying the LTL specification. In the event that the conjecture is not true, we wish to identify fragments of LTL for which the conjecture holds. Future research will also explore other application areas such as human-in-the-loop semiautonomous driving.




\appendix
  \setlength\abovedisplayskip{3pt}
  \setlength\belowdisplayskip{3pt}
  \setlength\abovedisplayshortskip{3pt}
  \setlength\belowdisplayshortskip{3pt}
\subsection{Proof of Theorem~\ref{thm:1}}
\label{appendix}

\begin{proof}
Suppose $\pbar$ satisfies $\phi$ with probability $1$, then the set of states of $M_{\Pro, \pbar}$ written as $MC_{\pbar}$ can be represented as a disjoint union of $T_{\pbar}$ transient states and $R^j_{\pbar}$ closed irreducible sets of recurrent classes~\cite{Durrett2012}:
\begin{equation} 
MC_{\pbar} = T_{\pbar} \sqcup R_{\pbar}^1 \sqcup \dotsc \sqcup R_{\pbar}^n
\end{equation}

\begin{prop}
\label{prop:rec}
Policy $\pbar$ satisfies $\phi$ with probability $1$ if and only if there exits $(\good,\bad) \in F_\Pro$ such that $\bad \in \trans$ and $\recc^j \cap \good \neq  \emptyset$ for all recurrent classes $\recc^j$.
\end{prop}

We omit the proof of Proposition~\ref{prop:rec}; however, it readily follows Definition~\ref{def:accept}.

Let $\Pi^*$ be the finite set of optimal policies that optimize the expected future utility.  We constructively show that for large enough values of $\gamma$, the discount factor and $w_B$, the negative reward on non accepting states, all policies $\pi^* \in \Pi^*$ satisfy $\phi$ with probability $1$.


Suppose $\pi^* \in \Pi^*$ does not satisfy $\phi$. Then one of the following two cases must be true:

\begin{itemize}
\item {\bf Case 1:} There exists a recurrent class $\recco^j$ such that $\recco^j \cap \good = \emptyset$. This means with policy $\pi^*$ it is possible to visit $\good$ only finitely often.
\item {\bf Case 2:} There exists $b \in \bad$ such that $b$ is recurrent. That is for some recurrent class of the $M_{\Pro, \pi^*}$,  $b \in \recco^j$. This translates to the possibility of visiting a state in $\bad$ infinitely often.
\end{itemize}

 We let $\Pi^* = \Pi_1 \cup \Pi_2$, where $\Pi_1 (\Pi_2)$ is the set of optimal policies that do not satisfy $\phi$ by violating {\bf Case 1} ({\bf Case 2}). Notice that this is not a disjoint union.

In addition, we know that the vector of utilities for any policy $\pi^* \in \Pi^*$ is ${\bf U}_{\pi^*} \in \mathbb{R}^N$, where $N = |MC_{\pi^*}|$ is the number of states of $M_{\Pro, \pi^*}$:

\begin{equation}
\label{eq:utility}
\begin{array}{ll}
&{\bf U}_{\pi^*} = \sum_{n=0}^{\infty} \gamma^n P_{\pi^*}^n {\bf W}   
\end{array}
\end{equation}

In this equation ${\bf U}_{\pi^*} = [U_{\pi^*}(s_0) \dotsc U_{\pi^*}(s_N)]^\top$ and ${\bf W} = [W(s_0) \dotsc W(s_N)]^\top$ and $P_{\pi^*}$ is the transition probability matrix with entries $p_{\pi^*}(s_i,s_j)$ which are the probability of transitioning from $s_i$ to $s_j$ using policy $\pi^*$.

We partition the vectors in equation~\eqref{eq:utility} into its transient and recurrent classes:

\begin{equation}
\label{eq:matrixutility}
\begin{bmatrix}
\Utb_{\pi^*}\\
\Ur_{\pi^*}
\end{bmatrix} = \sum_{n=0}^{\infty} \gamma^n
\begin{bmatrix}
P_{\pi^*}(T,T) & [P_{\pi^*}^{tr_1} \dotsc  P^{tr_m}_{\pi^*}]\\
{\bf 0}_{(\sum_{i=1}^mN_i\times q)} & P_{\pi^*}(R,R)
\end{bmatrix}^n
\begin{bmatrix}
\Wtb \\
\Wr
\end{bmatrix}
\end{equation}

In equation~\eqref{eq:matrixutility}, $\Utb_{\pi^*}$ is a vector representing the utility of every transient state. Assuming we have $q$ transient states, $P_{\pi^*}(T,T)$ is a $q\times q$ probability transition matrix containing the probability of transitioning from one transient state to another. 
Assuming there are $m$ different recurrent classes, ${\bf 0}_{(\sum_{i=1}^mN_i\times q)}$ is a zero matrix representing the probability of transitioning from any of the $m$ recurrent classes, each with size $N_i$ to any of the transient states. This probability is equal to $0$ for all of these entries.

On the other hand, ${\bf P_{\pi^*}} = [ P_{\pi^*}^{tr_1}\dotsc P_{\pi^*}^{tr_m}]$ is a $q \times \sum_{i=1}^m N_i$ matrix, where each $P_{\pi^*}^{tr_k}$ is a $q\times N_k$ matrix whose elements denote the probability of transitioning from any transient state $t_j$, $j\in \{ 1,\dotsc, q\}$ to every state of the $k$th recurrent class $R_{\pi^*}^k$.

Finally, $P_{\pi^*}(R,R)$ is a block diagonal matrix with $m$ blocks of size $\sum_{i=1}^m N_i \times \sum_{i=1}^m N_i$ for every recurrent class that states the probabilities of transitioning from one recurrent state to another. It is clear that $P_{\pi^*}(R,R)$ is a stochastic matrix since each block of $N_i \times N_i$ is a stochastic matrix~\cite{Durrett2012}.
From equation~\eqref{eq:matrixutility}, we can conclude:

\begin{align}
\label{eq:ur}
\Ur_{\pi^*} & = \sum_{n=0}^\infty \gamma^n 
\begin{bmatrix}
{\bf 0} & P_{\pi^*}(R,R)^n
\end{bmatrix}
\begin{bmatrix}
\Wtb \\
\Wr
\end{bmatrix}  \\
&= \sum_{n=0}^\infty \gamma^n P_{\pi^*}^n(R,R) \Wr 
 \end{align}

Also with some approximations, a lower bound on $\Utb_{\pi^*}$ can be found:

\begin{align}
\label{eq:ut}
&\sum_{n=0}^\infty \gamma^n
\begin{bmatrix}
P_{\pi^*}^n(T,T) & {\bf P}_{\pi^*} P_{\pi^*}^n(R,R)
\end{bmatrix}
\begin{bmatrix}
\Wtb \\
\Wr
\end{bmatrix} <  \Utb_{\pi^*}  \\
&\sum_{n=0}^{\infty} \gamma^n P_{\pi^*}^n(T,T) \Wtb + \sum_{n=0}^\infty \gamma^n {\bf P}_{\pi^*}P_{\pi^*}^n(R,R)\Wr < \Utb_{\pi^*} 
\end{align}

\noindent {\bf Case 1:}\\
We first consider all policies $\pi^* \in \Pi_1$. These are policies that violate case $1$, thus for $\pi^*$ there exists some $j$ such that $\recco^j \cap \good = \emptyset$. We choose any state $s \in \recco^j$.
Then we use equation~\eqref{eq:ur} to show that any policy $\pi^*$ over state $s$ has a non-positive utility $U_{\pi^*}(s) \leq 0$.

In equation~\eqref{eq:up0}, $k_1 = \sum_{j=0}^{i-1}N_j$, $k_2 = \sum_{j=i+1}^mN_j$, ${\bf p}^{rr_i}_{\pi^*}$ is the vector that corresponds to transition probabilities from $s \in \recco^j$ to any other state in the same recurrent class using policy $\pi^*$.  ${\bf W}_j = [W(s_1^j) \dotsc W(s_{Nj}^j)]$ is the vector for the reward values of the recurrent class $\recco^j$. Since none of these states are in $\good$, we conclude that for all elements $ w \in {\bf W}_j, \:w \leq 0$. 

\begin{align}
\label{eq:up0}
U_{\pi^*}(s) =& \Urso = \sum_{n=0}^\infty \gamma^n
\begin{bmatrix}
{\bf 0}_{k_1\times q} & {\bf p}_{\pi^*}^{rr_j}&{\bf 0}_{k_2\times q}
\end{bmatrix} \Wr \\
=&\sum_{n=0}^\infty \gamma^n {\bf p}_{\pi^*}^{rr_j} {\bf W}_j \leq 0 \implies U_{\pi^*}(s) \leq 0
\end{align}
We first consider the case that $s$ is in a recurrent class of $MC_{\pbar}$.
\begin{itemize}
\setlength{\leftmargin}{0pt}
\item If $s$ is in some recurrent class $s \in \recc^j$, by proposition~\ref{prop:rec}, $\recc^j \cap \good \neq \emptyset$. Therefore, there is at least one $s_g \in \good$ such that $s_g \in \recc^j$ and $s\in \recc^j$.
In addition, we know that all states in $\bad$ are in the transient class. Therefore the vector of rewards in this recurrent class ${\bf W}_j$ as defined previously contains non-negative elements. That is for all elements $w \in {\bf W}_j,\: 0 \leq w$ and there exists at least one $ w_g \in {\bf W}_j, \: 0<w_g$.

\begin{align}
\label{eq:upbar}
0<\sum_{n=0}^\infty \gamma^n {\bf p}^{rr_j}_{\bar{\pi}} {\bf W}_j 
 \implies 0 < U_{\bar{\pi}}(s)
\end{align}

We have shown that for some $s$, and any policy $\pi^* \in \Pi_1$, $\Uo(s) < \Ubar(s)$ which contradicts the optimality assumption of $\pi^*$ for the case where $s \in \recc^j$. Thus, we must have that $s$ is in a transient class of $MC_{\pbar}$. 

\item If $s$ is in a transient class $s\in \trans$, we first find a lower bound on $\Ut_{\bar{\pi}}(s)$, and show this lower bound can be greater than any positive number for large enough choice of $\gamma$. 
Note that at minimum all the states in the transient set of $\bar{\pi}$ will have utility of $w_B <0$, that is ${\bf W}^{\text{trans}} =  {\bf W}_B=[w_B \dotsc w_B]$, and there will be only one state $s_g \in \good$ that lives in the recurrent class. That is $w_G \in \Wr$ has a positive reward. 

\begin{prop}
\label{prop:N}
For transient states $t_1, t_2 \in T$,  there exists $N < \infty$ such that:
\begin{equation}
 \sum_{n=0}^{\infty} p^n(t_1,t_2) < N,
 \end{equation}
that is, the infinite sum is bounded~\cite{Durrett2012}. 
\end{prop}

We assume $\mathfrak{q} : = |T_{\pi}|$ is the number of transient states.

In addition, $P_{\pbar}^n(R,R)$ is a stochastic matrix with row sum of $1$~\cite{Durrett2012}.

\begin{align}
\label{eq:lowerbound}
 &\sum_{n = 0}^\infty \gamma^n P_{\pbar}^n(T,T){\bf W}^{\text{tr}} +
 \gamma^n {\bf P}_{\pbar}P_{\pbar}^n(R,R) \Wr < {\bf U}_{\bar{\pi}}^{\text{tr}} \\
 &N_1 \mathbb{I}_{\mathfrak{q} \times \mathfrak{q}}{\bf W}_B  + 
 \sum_{n=0}^\infty \gamma^n {\bf P}_{\pbar}P_{\pbar}^n(R,R) \Wr  < {\bf U}_{\bar{\pi}}^{\text{tr}}
 \end{align}



\begin{prop}
\label{prop:bound}
If $p^n(s,s)$ is the probability of returning from a state $s$ to itself in $n$ time steps, there exists a lower bound on $\sum_{n=0}^\infty \gamma^n p^n(s,s)$.\\ 
First, there exists $\bar{n}$ such that $p^{\bar{n}}(s,s)$ is nonzero and bounded. That is $s$ visits itself after $\bar{n}$ time steps with a nonzero probability. \\
Also we know
$(p^{\bar{n}}(s,s))^n < p^{n\bar{n}}(s,s) $. Therefore:
\begin{align}
 \sum_{n=0}^\infty \gamma^n p^n(s,s) & >\sum_{n=0}^\infty \gamma^{n\bar{n}}p^{n\bar{n}}(s,s) \\
& >\sum_{n=0}^\infty (\gamma^{\bar{n}})^n (p^{\bar{n}}(s,s))^n  \\
&>\frac{1}{1-\gamma^{\bar{n}}} \bar{p} 
\end{align}
\end{prop}

 
 
 Going back to equation~\eqref{eq:lowerbound}, we find a stricter lower bound on the utility of every state ${\bf U}_{\bar{\pi}}^{\text{tr}}(s)$ using proposition~\ref{prop:bound}:
 
\begin{align}
\label{eq:case1}
& N_1w_B +  \frac{1}{1-\gamma^{\bar{n}}} \bar{m} < U_{\bar{\pi}} (s) = {\bf U}_{\bar{\pi}}^{\text{tr}} (s) \\
\label{eq:case1b}
&\text{If } 0<N_1w_B + \frac{1}{1-\gamma^{\bar{n}}} \bar{m} \\
\label{eq:case1c}
&\implies
U_{\pi^*}(s) < U_{\pbar}(s)
\end{align}
 
 Here $\bar{m} = \max (\bar{M})$ and $\bar{M} < {\bf P}_{\pbar}\bar{P}\Wr$, where $\bar{P}$ is a block matrix whose nonzero elements are $\bar{p}$ bounds derived from proposition~\ref{prop:bound}.


For a fixed $w_B$, we can select a large enough $\gamma$ so equation~\eqref{eq:case1b} holds for all $\pi^* \in \Pi_1$.
This condition implies equation~\eqref{eq:case1c} which contradicts with optimality of any $\pi^* \in \Pi_1$. Therefore, $\pi^*$ cannot be optimal unless it visits $\good$ infinitely often.
\end{itemize}

\noindent {\bf Case 2:}\\
Now we consider case $2$, where $\pi^* \in \Pi_2$.  Here for some $b \in \bad$, $b \in \recco^j $. In addition, this state is in the transient class of $\pbar$, $b \in \trans$.
Using the same procedure as the previous case, we find the following upper bound.

\begin{align}
{\bf U}_{\bar{\pi}}^{\text{tr}} &>\sum_{n = 0}^\infty \gamma^n P_{\pbar}^n(T,T){\bf W}^{\text{tr}}  \\
&>\sum_{n=0}^\infty  P_{\pbar}^n (T,T){\bf W}^{\text{tr}} \\
 \text{ (Proposition~\ref{prop:N})}\quad \quad &> N_2{\mathbb I}_{\mathfrak{q}\times \mathfrak{q}} {\bf W}_B\\
\implies & 
{\bf U}_{\bar{\pi}} (b)  > N_2w_B 
\end{align}

We know that $b$ is in the recurrent class while using policy $\pi^*$. So we can use equation~\eqref{eq:ur} to find a bound on the utility. An upper bound assumes that all the other states in the recurrent class have positive reward of $w_G$.

\begin{align}
\label{eq:ur1}
&\Ur_{\pi^*} 
 = \sum_{n=0}^\infty \gamma^n P_{\pi^*}^n(R,R) \Wr \implies \\
 &U^{\text{rec}}_{\pi^*}(b) \leq \sum_{n=0}^\infty \gamma^n w_G  + \sum_{n=0}^\infty \gamma^n p_{\pi^*}^n(b,b) w_B \\
 &< w_G \frac{1}{1-\gamma}  + w_B\sum_{n=0}^\infty \gamma^n p_{\pi^*}^n(b,b)
 \end{align}

If the following condition in equation~\eqref{eq:optcondition} holds, we conclude that for a state $b$,  $U_{\pi^*}(b) < U_{\bar{\pi}}(b)$ which violates the optimality of $\pi^*$.

\begin{align}
\label{eq:optcondition}
&U_{\pi^*}(b) <  w_G \frac{1}{1-\gamma}+ w_B\sum_{n=0}^\infty \gamma^n p_{\pi^*}^n(b,b) < N_2w_B <U_{\bar{\pi}}(b)
\end{align}

We only need to enforce:
\begin{align}
\label{eq:optcondition2}
&w_G\frac{1}{1-\gamma}  + w_B\sum_{n=0}^\infty \gamma^n p_{\pi^*}^n(b,b) < N_2w_B 
\end{align}

Since there are only a finite number of policies in $\Pi_2$, from all policies $\pi^* \in \Pi_2$, we can find $\bar{p}$ such that:
\begin{equation}
\sum_{n=0}^\infty \gamma^n p_{\pi^*}^n(b,b) < \sum_{n=0}^\infty \gamma^n \bar{p}
\end{equation}

Therefore equation~\eqref{eq:optcondition2} can be simplified:
\begin{align}
\label{eq:optcondition3}
&w_G \frac{1}{1-\gamma} + w_B \sum_{n=0}^\infty \gamma^n \bar{p} < N_2w_B\\
&w_G \frac{1}{1-\gamma} + w_B \frac{1}{1- \gamma}\bar{p} < N_2w_B \\
&(w_G + w_B \bar{p})(\frac{1}{1-\gamma}) < N_2 w_B \\
\label{eq:optcondition3d}
&(w_G + w_B \bar{p}) - N_2 w_B(1-\gamma) < 0
\end{align}


We assumed without loss of generality $w_G =1$. For a fixed value of $\gamma$, we choose $w_B$ small enough so all $\pi^* \in \Pi_2$ satisfy equation~\eqref{eq:optcondition3d} and violate the optimality condition.

As a result, any optimal policy must satisfy case $2$, which is visiting a state in $\bad$ only finitely often.

For optimal policies $\pi^* \in \Pi_1 \cap \Pi_2$, we need to find $\gamma$ and $w_B$ such that both conditions for case 1 and case 2 are satisfied. That is:

\begin{equation}
\label{eq:final}
\begin{cases}
0<N_1w_B ( 1-\gamma ^ {\bar n}) +  \bar{M} \\
(1 + w_B \bar{p}) - N_2 w_B(1-\gamma) < 0
\end{cases}
\end{equation}

We select a pair of $\gamma$ and $w_B$ so the system of equations in~\eqref{eq:final} is satisfied. This solution can be found as follows:

First, for a small real number $0<\epsilon < \bar{M}$, we select $w_B^*$ so:

\begin{equation}
1 + w_B^* \bar{p} < -\epsilon
\end{equation}

Then, $\gamma^*$ is selected so the following holds:

\begin{equation}
\max \{ -N_1 w_B^* (1 - (\gamma^*)^{\bar{n}}) , -N_2 w_B^* (1 -\gamma^*) \} < \epsilon 
\end{equation}
The pair of $(w_B^*, \gamma^*)$ satisfy equation~\eqref{eq:final}, and as a result none of the policies $\pi^* \in \Pi^*$ are optimal.
\end{proof}


\small
\bibliographystyle{IEEEtran}
\bibliography{refs}

\end{document}